\documentclass[prodmode,acmec]{ec-acmsmall} 

\acmVolume{X}
\acmNumber{X}
\acmArticle{X}
\acmYear{2014}
\acmMonth{2}
\usepackage[numbers]{natbib}

\usepackage{amsfonts}
\usepackage{amsmath}
\usepackage{amssymb}
\usepackage{xspace}

\usepackage{todonotes}

\usepackage[overload]{textcase}

\newcommand{\nats}{\mathbb N}

\newcommand{\x}{\mathbf{x}}
\newcommand{\y}{\mathbf{y}}
\newcommand{\row}{\mathbf{r}}
\newcommand{\col}{\mathbf{c}}
\newcommand{\G}{\mathcal{G}}

\DeclareMathOperator{\supp}{Supp}
\DeclareMathOperator{\pr}{Pr}

\begin{document}

\markboth{J. Fearnley and R. Savani}{Finding Approximate Nash Equilibria of
Bimatrix Games via Payoff Queries}

\title{Finding Approximate Nash Equilibria of Bimatrix Games via Payoff Queries}

\author{John Fearnley 
\affil{University of Liverpool}
Rahul Savani
\affil{University of Liverpool}
}

\begin{abstract}
We study the deterministic and randomized query complexity of finding
approximate equilibria in bimatrix games. We show that the deterministic query
complexity of finding an $\epsilon$-Nash equilibrium when $\epsilon <
\frac{1}{2}$ is $\Omega(k^2)$, even in zero-one constant-sum games. In
combination with previous results~\cite{FGGS13}, this provides a complete
characterization of the deterministic query complexity of approximate Nash
equilibria. We also study randomized querying algorithms. We give a randomized
algorithm for finding a $(\frac{3 - \sqrt{5}}{2} + \epsilon)$-Nash equilibrium
using $O(\frac{k \cdot \log k}{\epsilon^2})$ payoff queries, which shows that
the $\frac{1}{2}$ barrier for deterministic algorithms can be broken by
randomization. For well-supported Nash equilibria (WSNE), we first give a
randomized algorithm for finding an $\epsilon$-WSNE of a zero-sum bimatrix game
using $O(\frac{k \cdot \log k}{\epsilon^4})$ payoff queries, and we then use
this to obtain a randomized algorithm for finding a $(\frac{2}{3} +
\epsilon)$-WSNE in a general bimatrix game using $O(\frac{k \cdot \log
k}{\epsilon^4})$ payoff queries. Finally, we initiate the study of lower bounds
against randomized algorithms in the context of bimatrix games, by showing that
randomized algorithms require $\Omega(k^2)$ payoff queries in order to find a
$\frac{1}{6k}$-Nash equilibrium, even in zero-one constant-sum games. In
particular, this rules out query-efficient randomized algorithms for finding
exact Nash equilibria.
\end{abstract}




\acmformat{J. Fearnley and R. Savani, 2014. Finding Approximate Nash Equilibria
of Bimatrix Games via Payoff Queries.}

\begin{bottomstuff}
This work is supported by EPSRC grant EP/H046623/1 ``Synthesis and Verification
in Markov Game Structures'' and EPSRC grant EP/L011018/1 ``Algorithms for
Finding Approximate Nash Equilibria.''
\end{bottomstuff}

\maketitle

\section{Introduction}

The Nash equilibrium is the central solution concept in game theory, which makes
algorithms for finding Nash equilibria an important topic. A recent strand of
work~\cite{B13b, BB13, FGGS13, GR13, HN13} has studied this problem from the
perspective of \emph{payoff query complexity}. 

The payoff query model is motivated by practical applications of game theory. In
many practical applications it is often the case that we know that there is a
game to be solved, but we do not know what the payoffs are. In order to discover
the payoffs, we would have to play the game, and hence discover the payoffs
through experimentation. This may be quite costly, so it is natural
to ask whether we can find a Nash equilibrium of a game while minimising the
number of experiments that we must perform.

Payoff queries model this situation. In the payoff query model we are told the
structure of the game, ie.\ the strategy space, but we are not told the
payoffs. We are permitted to make payoff queries, where we propose a pure
strategy profile, and we are told the payoff to each player under that strategy
profile. Our task is to compute an equilibrium of the game while minimising the
number of payoff queries that we make.

This paper studies the payoff query complexity of finding equilibria in bimatrix
games. Previous work has shown that the deterministic query complexity of
finding an \emph{exact} Nash equilibrium in a $k \times k$ bimatrix game is
$k^2$, even for zero-one constant-sum games~\cite{FGGS13}. In other words, one
cannot hope to find an exact Nash equilibrium without querying all pure strategy
profiles. Since query-efficient algorithms for exact equilibria do not exist,
this naturally raises the question: what is the payoff query complexity of
finding \emph{approximate} equilibria?

There are two competing notions of approximate equilibrium for bimatrix games.
An exact Nash equilibrium requires that all players achieve their best-response
payoff, and thus have no incentive to deviate. An \emph{$\epsilon$-Nash
equilibrium} requires that each player receives a payoff that is within
$\epsilon$ of their best-response payoff, and thus all players have only a
small incentive to deviate. One problem with this
definition is that an $\epsilon$-Nash equilibrium permits a player to place a
small amount of probability on a very bad strategy, and it is questionable
whether a rational player would actually want to do this. An
\emph{$\epsilon$-well supported Nash equilibrium}  ($\epsilon$-WSNE) rectifies this, by
requiring that all players only place probability on strategies that are within
$\epsilon$ of being a best-response.

Previous work by Fearnley, Gairing, Goldberg, and Savani~\cite{FGGS13} has shown
that an algorithm of Daskalakis, Mehta and Papadimitriou~\cite{DMP09} can be
adapted to produce a deterministic algorithm that finds a $\frac{1}{2}$-Nash
equilibrium using $2k - 1$ payoff queries. The same paper also shows that, for
all $i$ in the range $2 \le i < k$, the deterministic payoff query complexity of
finding a $(1 - \frac{1}{i})$-Nash equilibrium is at least $k - i +1$.
Note, in particular, that
this implies that for every constant $\epsilon$ in the range $\frac{1}{2} \le
\epsilon < 1$, the deterministic query complexity of finding an $\epsilon$-Nash
equilibrium is $\Theta(k)$.

However, relatively little is known about the more interesting case of $\epsilon
< \frac{1}{2}$ . A lower bound of $\Omega(k \cdot \log k)$ has been shown
against deterministic algorithms that find a $O(\frac{1}{\log k})$-Nash
equilibrium~\cite{FGGS13}. For the special case of zero-sum games, Goldberg and
Roth have
shown that an $\epsilon$-Nash equilibrium of a zero-sum
bimatrix game can be found by a \emph{randomized} algorithm that uses $O(\frac{k
\cdot \log k}{\epsilon^2})$ many payoff queries~\cite{GR13}.

\paragraph{Our contribution}

As we have mentioned, so far relatively little is known for the deterministic
query complexity of finding an $\epsilon$-Nash equilibrium when $\epsilon <
\frac{1}{2}$. We address this with a lower bound: in
Section~\ref{sec:epsdlower} we show that, for every $\epsilon > 0$, the
deterministic payoff query complexity of finding a $(\frac{1}{2} -
\epsilon)$-Nash equilibrium in a $k \times k$ bimatrix game is $\Omega(k^2)$.
Our lower bound holds even for zero-one constant-sum games. 
When combined with
previous results, this provides a complete characterization of the
deterministic query complexity of $\epsilon$-Nash equilibria when $\epsilon$ is
constant: it is $\Theta(k)$ for $\epsilon \ge \frac{1}{2}$, and $\Theta(k^2)$
for $\epsilon < \frac{1}{2}$.

Since our lower bound rules out deterministic query-efficient algorithms for
finding $(\frac{1}{2} - \epsilon)$-Nash equilibria, it is natural to ask whether
this threshold can be broken through the use of randomization. In
Section~\ref{sec:enbm} we give a positive answer to this question. Our approach
is to take an algorithm of Bosse, Byrka, and Markakis~\cite{BBM10}, and to apply
the randomized algorithm of Goldberg and Roth~\cite{GR13} in order to solve the
zero-sum game used by the BBM algorithm. We show that this produces a randomized
algorithm that uses $O(\frac{k \cdot \log k}{\epsilon^2})$ payoff queries and
finds a $(\frac{3 - \sqrt{5}}{2} + \epsilon)$-Nash equilibrium in a $k \times k$
bimatrix game. Since $\frac{3 - \sqrt{5}}{2} \approx 0.382$, this shows that
randomization can be used to defeat the barrier at $0.5$ that exists for
deterministic algorithms.

In Section~\ref{sec:wsnebm}, we turn our attention to approximate well-supported
Nash equilibria, which is a topic that has not previously been studied from the
payoff query perspective in the context of bimatrix games. We adapt the
algorithm of Kontogiannis and Spirakis for finding a $\frac{2}{3}$-WSNE in a
bimatrix game~\cite{KS10}, and in doing so we obtain a randomized algorithm for
finding a $(\frac{2}{3} + \epsilon)$-WSNE using $O(\frac{k \cdot \log
k}{\epsilon^4})$ many payoff queries. Note that this almost matches the best
known polynomial-time algorithm for finding approximate well-supported Nash
equilibria: Fearnley, Goldberg, Savani, and S\o rensen~\cite{FGSS12} have given
a polynomial time algorithm that finds a $(\frac{2}{3} - 0.004735)$-WSNE.

As part of the proof of this result, we prove an interesting side result: there
is a randomized algorithm that finds an $\epsilon$-WSNE of a zero-sum game using
$O(\frac{k \cdot \log k}{\epsilon^4})$ payoff queries. This complements the
result of Goldberg and Roth for finding $\epsilon$-Nash equilibria in zero-sum
games, and potentially could find applications elsewhere.

Finally, we initiate the study of lower bounds against randomized algorithms in
the context of bimatrix games. In Section~\ref{sec:rlower} we show that the
randomized query complexity of finding a $\frac{1}{6k}$-Nash equilibrium in a $k
\times k$ bimatrix game is $\Omega(k^2)$.
Again, our lower bound holds even for zero-one constant-sum games.
Note, in particular, that our lower bound rules out the existence of query-efficient
randomized algorithms for finding exact equilibria. This improves over earlier
results, which only considered deterministic algorithms for finding exact Nash
equilibria~\cite{FGGS13}.

\paragraph{Related work}

Apart from the related work on bimatrix games that we have already mentioned,
there have been a number of results on payoff query complexity in the context of
$n$-player strategic form games.
Recently, Babichenko has shown that there is a constant (but small) $\epsilon$
for which the randomized query complexity of finding an
$\epsilon$-well-supported Nash equilibrium in an $n$-player strategic form is
exponential in~$n$~\cite{B13b}.

The query complexity of finding approximate-correlated equilibria has been
studied in a pair of papers~\cite{BB13,HN13}, where it was shown that a
randomized algorithm can find approximate-correlated equilibria in polynomial
time, but that non-random algorithms require exponentially many queries, and
that finding an exact equilibrium requires exponentially many queries. 

Finally, Goldberg and Roth have given a randomized algorithm that uses
logarithmically many payoff queries in order to find approximate correlated
equilibria in binary-action $n$-player strategic form games, and they also give
a matching lower bound~\cite{GR13}. The same paper shows a linear lower bound
for finding well-supported approximate correlated Nash equilibria.

Of course, there has been much previous work studying approximate Nash
equilibria from the computational complexity point of
view~\cite{BBM10,DMP07,DMP09,KPS09,TS08}, some of which we have drawn upon for
our query complexity results. So far, the best polynomial-time algorithm for
finding $\epsilon$-Nash equilibria was given by Tsaknakis and Spirakis, who
showed that a $0.3393$-Nash equilibrium can be found in
polynomial-time~\cite{TS08}. For
well-supported approximate Nash equilibria, the first result on the subject gave
an algorithm for finding a $\frac{5}{6}$-WSNE in polynomial time~\cite{DMP09},
but this only holds if a certain unproved graph-theoretic conjecture is true.
The best result until recently was by Kontogiannis and Spirakis, who gave an
algorithm for finding $\frac{2}{3}$-WSNE in polynomial time. The current best
known algorithm was given by Fearnley, Goldberg, Savani, and S\o
rensen~\cite{FGSS12} who, in a slight improvement over previous work, produced a
polynomial-time algorithm for finding a $(\frac{2}{3} - 0.004735)$-WSNE. 

There has also been a line of work studying the support size requirements for
approximate Nash equilibria.
It has been shown that every game has a $\frac{1}{2}$-Nash equilibrium with
support size $2$~\cite{DMP09}, but logarithmic support sizes are both
necessary~\cite{FNS07} and sufficient~\cite{A94,LMM03} for $\epsilon$-Nash
equilibria with $\epsilon < \frac{1}{2}$. The threshold of $\frac{1}{2}$, of
course, also appears in our work on the deterministic query complexity of
approximate Nash equilibria.

A similar support-size threshold may exist for well-supported Nash equilibria at
$\frac{2}{3}$: it has been shown that $\epsilon$-WSNE with $\epsilon <
\frac{2}{3}$ require super-constant support sizes~\cite{ANSV13}, whereas an
as yet unproved graph theoretic conjecture would imply that every game has a
$\frac{2}{3}$-WSNE with support size 3~\cite{DMP09}.





\section{Preliminaries}
\label{sec:prelim}

\paragraph{Games and Strategies}

A $k \times k$ \emph{bimatrix game} is a pair $(R, C)$ of two $k \times k$
matrices: $R$ gives payoffs for the \emph{row player}, and $C$ gives payoffs for
the \emph{column player}. We make the standard assumption that all payoffs lie
in the range $[0, 1]$. For each $n \in \nats$, we use $[n]$ to denote the set
$\{1, 2, \dots, n\}$. Each player has $k$ \emph{pure strategies}. To play the
game, both players simultaneously select a pure strategy: the row player selects
a row $i \in [k]$, and the column player selects a column $j \in [k]$. The row
player then receives payoff $R_{i, j}$, and the column player receives payoff
$C_{i, j}$. We say that a bimatrix game $(R, C)$ is a \emph{zero-one} game, if
all entries of $R$ and $C$ are either $0$ or $1$. We say that $(R, C)$ is
\emph{constant-sum} if there is a constant $c$ such that $R_{i,j} + C_{i,j} =
c$, for all $i,j \in [k]$. 

A \emph{mixed strategy} is a probability distribution over $[k]$. We denote
a mixed strategy for the row player as a row vector $\x$ of length $k$, such
that $\x_i$ is the probability that the row player assigns to pure
strategy~$i$. A mixed strategy of the column player is a column vector $\y$ of
length $k$, with the same interpretation. Given a mixed strategy $\x$ for
either player, the \emph{support} of $\x$, denoted $\supp(\x)$, is the set of
pure strategies $i$ with $x_i > 0$. If $\x$ and $\y$ are mixed strategies for
the row and column player, respectively, then we call $(\x, \y)$ a \emph{mixed
strategy profile}. 

\paragraph{Solution Concepts}

Let $(\x, \y)$ be a mixed strategy profile in a $k \times k$ bimatrix game $(R,
C)$. We say that a row $i \in [k]$ is a \emph{best response} for the row
player if $R_i \cdot \y = \max_{j \in [k]} R_j \cdot \y$.  We say that a column
$i \in [k]$ is a best response for the column player if $(\x \cdot C)_i =
\max_{j \in [k]} (\x
\cdot C)_j$. We define the \emph{regret} suffered by the row player to be the
difference between the payoff that the row player obtains under $(\x, \y)$, and
the payoff of a best response. More formally, the row player's regret is
$\max_{j \in [k]}( R_j \cdot \y) - \x \cdot R \cdot \y$.
Similarly, the column player's regret is defined to be
$\max_{j \in [k]}((\x \cdot C)_j) - \x \cdot C \cdot \y$.
The mixed strategy profile $(\x, \y)$ is a \emph{mixed Nash equilibrium} if
both players have regret $0$ under $(\x, \y)$. Note that this is equivalent to
saying that every pure strategy in $\supp(\x)$ is a best response against~$\y$,
and every pure strategy in $\supp(\y)$ is a best response against~$\x$. 

The two approximate solution concepts that we study in this paper both
weaken the requirements of a mixed Nash equilibrium, but in different ways. An
\emph{$\epsilon$-Nash equilibrium} is an approximate solution concept that
weakens the regret based definition of a mixed Nash equilibrium. For every
$\epsilon \in [0, 1]$, a mixed strategy profile $(\x, \y)$ is an
$\epsilon$-Nash equilibrium if both player suffer regret at most $\epsilon$
under $(\x, \y)$.

An \emph{$\epsilon$-well supported Nash equilibrium} ($\epsilon$-WSNE) weakens the
best-response definition of a mixed Nash equilibrium. A strategy $i \in
[k]$ is an $\epsilon$-best response against $\y$ if:
\begin{equation*}
R_i \cdot \y \ge \max_{j \in [k]}( R_j \cdot \y) - \epsilon.
\end{equation*}
Similarly, a strategy $j \in [k]$ is an $\epsilon$-best response against $\x$ if:
\begin{equation*}
(\x \cdot C)_j \ge \max_{j \in [k]}((\x \cdot C)_j) - \epsilon.
\end{equation*}
For every $\epsilon \in [0, 1]$, a mixed strategy profile $(\x, \y)$ is an
$\epsilon$-WSNE if every pure strategy in $\supp(\x)$ is an $\epsilon$-best
response against~$\y$, and every pure strategy in $\supp(\y)$ is an
$\epsilon$-best response against~$\x$. Note that every $\epsilon$-WSNE is an
$\epsilon$-Nash equilibrium, but that the converse does not hold.

\paragraph{Payoff Queries} 
In the payoff query model, an algorithm initially only knows the size of the
game, but does not know the payoffs. That is, the algorithm knows that $(R, C)$
is a $k \times k$ bimatrix game, but it does not know any of the payoffs in $R$
or $C$. In order to discover the payoffs, the algorithm must make payoff
queries. A payoff query is a pair $(i, j)$ where $i \in [k]$ is a pure strategy
of the row player, and $j \in [k]$ is a pure strategy of the column player. When
an algorithm makes a payoff query $(i, j)$, it receives a pair $(a, b)$ where $a
= R_{i, j}$ is the row player payoff and $b = C_{i,j}$ is the column player
payoff.

\section{An $\Omega(k^2)$ lower bound against deterministic algorithms for finding $(\frac{1}{2} - \epsilon)$-Nash equilibria}
\label{sec:epsdlower}

In this section, we show that for every $\epsilon > 0$, no deterministic
algorithm can find a $(\frac{1}{2} - \epsilon)$-Nash equilibrium in a $k \times
k$ bimatrix with fewer than $\frac{\epsilon}{2} \cdot k^2$ queries. Thus, we
show that the deterministic query complexity of finding a $(\frac{1}{2} -
\epsilon)$-Nash equilibrium is $\Omega(k^2)$ for all $\epsilon > 0$. This lower
bound holds even for zero-one constant-sum games.

Our proof will take the form of an algorithm interacting with an
adversary, who will respond to the payoff queries that are made by the
algorithm. The fundamental idea behind our lower bound is that the adversary
will hide a column $c$ such that $C_{i, c} = 1$ for all $i \in [k]$. Thus,
column $c$ always has payoff $1$ for the column player, no matter what strategy
the row player is using. Our goal is to show that the column
player must place a significant amount of probability on $c$ in order to be in a
$(\frac{1}{2} - \epsilon)$-Nash equilibrium.

We now define our adversary strategy. For each payoff query $(i, j)$, we
respond according to the following rules:
\begin{itemize}
\item If column $j$ has received fewer than $\epsilon \cdot k$ payoff queries,
then we respond with $(0, 1)$.
\item If column $j$ has received at least $\epsilon \cdot k$ payoff queries,
then we respond with $(1, 0)$.
\end{itemize}
The result of the interaction between the algorithm and the adversary is a
\emph{partial bimatrix game}, which is a bimatrix game $(R, C)$ where $R_{i,j}$
and $C_{i,j}$ are defined only for the pairs $(i, j)$ that have received a
payoff query. 

The idea behind this strategy is that, if an algorithm makes fewer than
$\epsilon \cdot k$ payoff queries in a column~$j$, then it cannot rule out the
possibility that $j$ is the hidden column~$c$. Crucially, if an algorithm makes
fewer than $\frac{\epsilon}{2} \cdot k^2$ payoff queries overall, then there will
be at least $\frac{k}{2}$ columns that receive fewer than $\epsilon \cdot k$
payoff queries, and thus, there will be at least $\frac{k}{2}$ possible
candidates for the hidden column.

So, suppose that our algorithm made fewer than $\frac{\epsilon}{2} \cdot k^2$
payoff queries and then produced a mixed strategy profile $(\x, \y)$. Let $(R',
C')$ be the resulting partial bimatrix game. We will extend $(R', C')$ to a
fully defined bimatrix game $(R, C)$ as follows:
\begin{itemize}
\item We first place the hidden column. To do so, we pick a column $c$ with
$\y_c < \frac{2}{k}$, such that $c$ has received fewer than $\epsilon \cdot k$
payoff queries. Since at least $\frac{k}{2}$ columns received
fewer than $\epsilon \cdot k$ payoff queries, such a column is guaranteed to
exist. We set $R_{i,c} = 0$ and $C_{i, c} = 1$ for all $i \in [k]$.
\item For each column $j \ne c$, we set all unqueried elements of $j$ to be $(1,
0)$. More formally, for each column $j \ne c$ and each $i \in [k]$, we set
\begin{equation*}
R_{i, j} = 
\begin{cases}
R'_{i,j} & \text{if $R'_{i,j}$ is defined.} \\
1 & \text{otherwise,}
\end{cases}
\end{equation*}
and we set $C_{i,j} = 1 - R_{i,j}$.
\end{itemize}
Note that $(R, C)$ is a zero-one constant-sum bimatrix game. The mixed strategy
profile $(\x, \y)$ and the bimatrix game $(R, C)$ will be fixed for the rest of
this section. Observe that, by construction, we have ensured that $\y$ plays $c$
with low probability. We will exploit this to show that $(\x, \y)$ is not a
$(\frac{1}{2} - \epsilon)$-Nash equilibrium in $(R, C)$. 

In our first lemma, we give a lower bound on the row player's best response
payoff against $\y$. 


\begin{lemma}
\label{lem:sgame}
The row player's best response payoff against $\y$ in $(R, C)$ is at least
$1 - \epsilon - \frac{2}{k}$.
\end{lemma}
\begin{proof}
Consider a strategy $\x'$ that mixes uniformly over all rows. Let $j$ be a column
such that $j \ne c$. By definition, we have that there are at most $\epsilon
\cdot k$ rows $i$ in $j$ that satisfy $C_{i, j} = 1$, and there are at least $k
- \epsilon \cdot k$ rows $i$ in $j$ that satisfy $C_{i, j} = 0$. Therefore,
since $\x'$ plays each row with probability $\frac{1}{k}$, we have $(\x' \cdot
C)_{j} \le \epsilon$. Note that $\y$ plays $c$ with probability at most
$\frac{2}{k}$. Therefore, we have the following bound on the payoff of $\y$
against $\x'$:
\begin{align*}
\x' \cdot C \cdot \y &= \y_c \cdot (\x' \cdot C)_c + \sum_{j \ne c} \y_j
\cdot (\x' \cdot C)_j \\
&\le \frac{2}{k} + \sum_{j \ne c} \y_j \cdot \epsilon \\
&\le \frac{2}{k} + \epsilon.
\end{align*}
Since the game is constant sum, this implies that the payoff of $\x'$ against
$\y$ is at least $1 - \epsilon - \frac{2}{k}$. Therefore, the row player's best
response against $\y$ must also have payoff at least $1 - \epsilon - \frac{2}{k}$.
\qed
\end{proof}

We now use the previous lemma to show that, provided that $\frac{2}{k} <
\epsilon$, the mixed strategy profile $(\x, \y)$ cannot be a
$(\frac{1}{2}-\epsilon)$-Nash equilibrium. 

\begin{lemma}
\label{lem:colreplace}
If $\frac{2}{k} < \epsilon$, then $(\x, \y)$ is not a
$(\frac{1}{2}-\epsilon)$-Nash equilibrium in $(R, C)$.
\end{lemma}
\begin{proof}
Suppose for the sake of contradiction that $(\x, \y)$ is a $(\frac{1}{2} -
\epsilon)$-Nash equilibrium in $(R, C)$. By Lemma~\ref{lem:sgame}, the row
player's best response payoff against $\y$ is at least $1 -
\epsilon - \frac{2}{k}$. Since, by assumption, the regret of the row player is
at most $\frac{1}{2} - \epsilon$, the payoff of $\x$ against $\y$ must be at
least: 
\begin{equation*}
\left(1 - \epsilon - \frac{2}{k}\right) - \left(\frac{1}{2} -\epsilon\right) =
\frac{1}{2} - \frac{2}{k}.
\end{equation*} 
Since the game is constant sum, this then
implies that the payoff of $\y$ against $\x$ is at most $\frac{1}{2} +
\frac{2}{k}$. But,
since the payoff of column $c$ is always $1$, the
payoff of $\y$ against $\x$ must be at least $1 - (\frac{1}{2} - \epsilon) =
\frac{1}{2} + \epsilon$. Thus, we have a contradiction whenever $\frac{2}{k} <
\epsilon$.
\qed
\end{proof}

Note that the precondition of Lemma~\ref{lem:colreplace} is equivalent to $k >
\frac{2}{\epsilon}$. Therefore this lemma shows that, for every $\epsilon > 0$,
there exists a $k'$ such that, for all $k \ge k'$ no deterministic algorithm can
find a $(\frac{1}{2} - \epsilon)$-Nash equilibrium in a $k \times k$ bimatrix
game with fewer than $\frac{\epsilon}{2} \cdot k$ payoff queries. Thus, we have
shown the following theorem, which is the main result of this section.

\begin{theorem}
The deterministic query complexity of finding a $(\frac{1}{2} - \epsilon)$-Nash
equilibrium is $\Omega(k^2)$, even in zero-one constant-sum games.
\end{theorem}

\section{An $\Omega(k^2)$ lower bound against randomized algorithms for finding
$\frac{1}{6k}$-Nash equilibria}
\label{sec:rlower}

In this section, we show that no randomized algorithm can find a
$\frac{1}{6k}$-Nash equilibrium while making $o(k^2)$ playoff queries, even in
zero-one constant-sum games.

At a high level, our technique is similar to the one we used for our
deterministic lower bound in in Section~\ref{sec:epsdlower}: we will hide a
column $c$ such that $C_{i, c} = 1$ for all $i \in [k]$. In this case however, 
instead of 
using an adversary, we use a probability distribution over games.
For each column $j \ne c$, we will select a single row
$r_j$ uniformly at random, and set $R_{r_j, j} = 1$ and $C_{r_j, j} = 0$. For
each row $i \ne r_j$ we will set $R_{i, j} = 0$ and $C_{i,j} = 1$. Thus, in
order to distinguish between column $c$, and a column $j \ne c$, the algorithm
must find the row $r_j$, and as we will show, even randomized algorithms cannot
do this in a query efficient manner.

Formally, we define $\G^k$ to be a random distribution over zero-one
constant-sum $k \times k$ bimatrix games, defined in the following way. To draw
a game from $\G^k$, we first choose a column $c \in [k]$ uniformly at random,
which will be referred to as the hidden column. Furthermore, we choose $r_1,
r_2, \dots, r_k$ to be $k$ uniformly and independently chosen rows (which may
include repeats). Then we construct the game $(R, C)$, where for all $i, j \in
[k]$ we have:
\begin{equation*}
C_{i,j} = \begin{cases}
1 & \text{if $j = c$,} \\
0 & \text{if $j \ne c$ and $r_j = i$,} \\
1 & \text{otherwise.}
\end{cases}
\end{equation*}
We define $R_{i,j} = 1 - C_{i,j}$ for all $i,j \in [k]$.

We show that if $(\x, \y)$ is a strategy profile in which $\y$
does not assign strictly more than $\frac{1}{2}$ probability to the hidden
column $c$, then $(\x, \y)$ is not a $\frac{1}{6k}$-Nash equilibrium. Since $\y$
can assign strictly more than $\frac{1}{2}$ probability to at most one column,
this implies that any algorithm that produces a $\frac{1}{6k}$-Nash equilibrium
must find the hidden column. The following lemma shows that this property holds
for all games in the support of the distribution $\G^k$.

\begin{lemma}
\label{lem:halfgame}
Let $(R, C)$ be a game drawn from $\G^k$, where $c$ is the hidden column. If
$(\x, \y)$ is a $\frac{1}{6k}$-Nash equilibrium in $(R, C)$, then $\y_c >
\frac{1}{2}$.
\end{lemma}
\begin{proof}
Suppose, for the sake of contradiction, that $(\x,\y)$ is a $\frac{1}{6k}$-Nash
equilibrium in $(R, C)$ and that $\y_c \le \frac{1}{2}$. We begin by arguing that
the row player's best response payoff against $\y$ is at least $\frac{1}{2k}$.
To see this, let $\x'$ be a row-player strategy that mixes uniformly over all
rows. Consider a column $j \ne c$, and observe that by construction, column $j$ contains
exactly one $0$ entry for the column player. Therefore, we have $(C \cdot
\x')_{j} = 1 - \frac{1}{k}$. Since $\y$ assigns at least $\frac{1}{2}$
probability to columns other than $c$, we have that the payoff of $\y$ against
$\x'$ is at most:
\begin{equation*}
\frac{1}{2} \cdot 1 + \frac{1}{2} \cdot (1 - \frac{1}{k}) = 1-\frac{1}{2k}. 
\end{equation*}
Since the game is constant-sum, this then
implies that the payoff of $\x'$ against $\y$ is at least $\frac{1}{2k}$. 

Since the row player's best response payoff against $\y$ must be at least
$\frac{1}{2k}$, and since $(\x, \y)$ is a $\frac{1}{6k}$-Nash equilibrium, we
have that the payoff of $\x$ against $\y$ is at least $\frac{1}{3k}$. Since the
game is constant-sum, this then implies that the payoff of $\y$ against $\x$ is
at most $1 - \frac{1}{3k}$. However, since the column player can obtain a payoff
of $1$ by playing column $c$, we have shown that the regret of $\y$ is strictly
more than $\frac{1}{6k}$, which provides our contradiction. \qed
\end{proof}

Recall that our goal is to show that even randomized algorithms cannot determine
if a column is the hidden column in a query efficient manner. In the next lemma,
we will formalise this idea. Suppose that our algorithm makes a series of
randomized queries. For each column $j$, we will use $A_j$ to be an indicator
variable for the event ``the algorithm did not find a row $i$ such that $C_{i,
j} = 0$''. Intuitively, if $A_j = 1$, then the algorithm cannot rule out the
possibility that column $j$ is the hidden column. The following lemma gives a
simple bound on the probability of $A_j$.

\begin{lemma}
\label{lem:littleo}
Let $(R, C)$ be a game drawn from $\G^k$.
Suppose that an algorithm makes $f(k)$ queries in column $j$. We have:
\begin{equation*}
\pr(A_j = 1) \ge 1 - \frac{f(k)}{k}.
\end{equation*}
\end{lemma}
\begin{proof}
If $j = c$, then $\pr(A_j = 1) = 1$, and the lemma holds.
Otherwise, since the row $r_j$ was chosen uniformly at random, if we make $f(k)$
queries in column $j$, then we find $r_j$ with probability $\frac{f(k)}{k}$.
Therefore $\pr(A_j = 1) = 1 - \frac{f(k)}{k}$. 
\qed
\end{proof}

Note that, if $f(k) \in o(k)$, then $1 - \frac{f(k)}{k}$ tends to $1$ as $k$
tends to infinity. Thus we will have $A_j = 1$ almost always, and $j$ will be
indistinguishable from $c$. In the following lemma, we use this fact to show our
lower bound. In particular, we show that no algorithm that makes $o(k^2)$ payoff
queries can succeed with a positive constant probability.

\begin{lemma}
\label{lem:rlowerb}
Let $(R, C)$ be a game drawn from $\G^k$.
Any algorithm that makes $o(k^2)$ payoff queries cannot find a
$\frac{1}{6k}$-Nash equilibrium of $(R, C)$ with positive constant probability. 
\end{lemma}
\begin{proof}
Suppose that the algorithm makes at most $g(k) \in o(k^2)$
payoff queries and outputs a mixed strategy profile $(\x, \y)$. Suppose, for the
sake of contradiction, that $(\x, \y)$ is a $\frac{1}{6k}$-Nash equilibrium with
probability $3 \cdot \delta$, for constant  $\delta \in (0, \frac{1}{3}].$ 

Note that at least $k - \delta \cdot k$ columns must receive $o(k)$ payoff
queries, because otherwise $\delta \cdot k$ columns would receive $\Omega(k)$
payoff queries, giving $\Omega(k^2)$ payoff queries in total. Thus, there must
exist a set $S$ of columns, with $|S| = k - \delta \cdot k$, and a function
$f(k) \in o(k)$, such that each column in $S$ received at most $f(k)$ payoff
queries. 

Let $S' \subseteq S$ be the set of columns in $S$ such that $A_j = 1$.
Recall that, for each column $j$ in $S'$, the algorithm cannot distinguish
between $j$ and $c$. We now prove a lower bound on the size of $S'$.
By Lemma~\ref{lem:littleo}, we have that if $j \in S$,
then $E[A_{j}] \ge 1 - \frac{f(k)}{k}$.
Define $A = \sum_{j
\in S} A_{j}$, and note that by linearity of expectations we have:
\begin{equation*}
E[A] \ge (k - \delta \cdot k) \cdot (1 - \frac{f(k)}{k})
\end{equation*}
Since each of the events corresponding to $A_{j}$ are independent, we can apply Hoeffding's
inequality to obtain:
\begin{equation*}
\pr(|A - E[A]| \ge \delta \cdot k) \le 2 \cdot \exp(-\frac{2 \cdot (\delta \cdot k)^2}{k -
\delta k}) 
\end{equation*}
Thus, with probability at least $1 - 2 \cdot \exp(-\frac{2 \cdot (\delta \cdot
k)^2}{k - \delta k})$ we have that $S'$ contains at least
\begin{equation*}
(k - \delta \cdot k) \cdot (1 - \frac{f(k)}{k}) - \delta \cdot k
\end{equation*}
columns. 

Let us focus on the case where $S'$ contains at least this many columns. Note
that $\y$ can assign strictly more than $\frac{1}{2}$ probability to at most one
column in $S'$, and therefore there are at least $|S'| - 1$ columns in $S'$ that
are not assigned strictly more than $\frac{1}{2}$ probability by $\y$. Since $c$
is chosen uniformly at random, and since the hidden column could be any of the
columns in $S'$, we have that $\y_c \le \frac{1}{2}$ with probability at least:
\begin{align*}
(|S'| - 1) \cdot \frac{1}{k} &= 
\frac{(k - \delta \cdot k) \cdot (1 - \frac{f(k)}{k}) - \delta \cdot k - 1}{k}\\
&=(1 - \delta) \cdot (1 - \frac{f(k)}{k}) - \delta  - \frac{1}{k}.
\end{align*}
By Lemma~\ref{lem:halfgame}, if $\y_c \le \frac{1}{2}$, then $(\x, \y)$ is not a
$\frac{1}{6k}$-Nash equilibrium.

In summary, we have that $(\x, \y)$ is not a $\frac{1}{6k}$-Nash
equilibrium with probability at least
\begin{equation*}
\left(1 - 2 \cdot \exp(-\frac{2 \cdot (\delta \cdot k)^2}{k - \delta k})\right)
\cdot \left((1 - \delta) \cdot (1 - \frac{f(k)}{k}) - \delta  -
\frac{1}{k}\right),
\end{equation*} 
where the first term is the probability that $S'$ contains enough columns, and
the second term is the probability that $\y_c \le \frac{1}{2}$. 
As $k$ tends to infinity, we have that 
$\exp(-\frac{2 \cdot (\delta \cdot k)^2}{k - \delta k})$ tends to $0$, and since
$f(k) \in o(k)$, we have that $\frac{f(k)}{k}$ tends to $0$. Thus, the entire
expression tends to $1 - 2 \cdot \delta$. Thus, for large $k$, the
algorithm will fail with probability strictly greater than $1 - 3 \cdot \delta$,
which contradicts our assumption that the algorithm succeeds with probability at
least $3 \cdot \delta$.
\end{proof}

Lemma~\ref{lem:rlowerb} shows that every algorithm that makes $o(k^2)$ payoff
queries on a game drawn from $\G^k$ will fail to find a $\frac{1}{6k}$-Nash
equilibrium with probability tending to 1 as $k$ tends to infinity. Moreover,
recall that all games in $\G^k$ are zero-one constant-sum games. Therefore,
we have shown the following theorem, which is the main result of this section.

\begin{theorem}
The randomized query complexity of finding a $\frac{1}{6k}$-Nash equilibrium is
$\Omega(k^2)$, even in zero-one constant-sum games.
\end{theorem}

\section{Zero-sum Games}
\label{sec:zero-sum}

We now turn our attention to showing upper bounds. In Sections~\ref{sec:enbm}
and~\ref{sec:wsnebm}, we will give query efficient randomized algorithms for
finding $(\frac{3 - \sqrt{5}}{2} + \epsilon)$-Nash equilibria, and $(\frac{2}{3}
+ \epsilon)$-WSNE, respectively. In this section, we give some preliminary
results on zero-sum games, which are required for our later results: the result
in Section~\ref{sec:enbm} requires us to find an $\epsilon$-Nash equilibrium of
a zero-sum game, and the result in Section~\ref{sec:wsnebm} requires us to find
an $\epsilon$-WSNE of a zero-sum game. In Section~\ref{sec:gr}, we present the
previous work of Goldberg and Roth~\cite{GR13}, which provides a randomized
query-efficient algorithm for finding an $\epsilon$-Nash equilibrium in a
zero-sum game, and in Section~\ref{sec:grwsne}, we show how this can be
converted into a randomized query-efficient algorithm for finding an
$\epsilon$-WSNE in a zero-sum game.

Recall that we assumed that all bimatrix game payoffs lie in the range $[0,1]$.
The two algorithms that we adapt both solve games that meet this assumption. In
doing so, both algorithms create and solve a derived zero-sum game with payoffs
lying in the range $[-1,1]$. Thus, for the sake of convenience, during this
section on zero-sum games, we assume that all payoffs lie in the range
$[-1,1]$.

\subsection{A randomized algorithm for finding an $\epsilon$-Nash equilibrium of
a zero-sum game}
\label{sec:gr}

The following theorem, shown by Goldberg and Roth~\cite{GR13} using using
multiplicative weights updates no-regret algorithms, states that we have a
randomized $O(\frac{k \cdot \log k}{\epsilon^2})$ payoff query algorithm that
finds an $\epsilon$-Nash equilibrium in a zero-sum game.

\begin{theorem}[\cite{GR13}]
\label{thm:zerosum}
An $\epsilon$-Nash equilibrium in a $k \times k$ zero-sum bimatrix game can,
with probability $1 - k^{-\frac{1}{8}}$, be computed using $O(\frac{k \cdot \log k}{\epsilon^2})$ payoff queries.
\end{theorem}

When we apply this result, we will also need to know the payoff vectors for both
players. We now give a randomized $O(\frac{k \cdot \log k}{\epsilon^2})$
algorithm for discovering \emph{approximate} payoff vectors. Let $(\x, \y)$ be a
mixed strategy profile in a $k \times k$ bimatrix game $(R, C)$. We say that a
$k$-dimensional vector $\row$ is an \emph{$\epsilon$-approximate payoff vector}
for the row player if, for each $i \in [k]$, we have $| \row_i - R_i \cdot \y |
\le \epsilon$. Similarly, we say that a $k$-dimensional vector $\col$ is an
$\epsilon$-approximate payoff vector for the column player if, for each $i \in
[k]$, we have $|\col_i - (\x \cdot C)_i| \le \epsilon$. The following lemma
shows that we can use a randomized algorithm to find an $\epsilon$-approximate
payoff vector for the row player using $O(\frac{k \cdot \log k}{\epsilon^2})$
payoff queries.

\begin{lemma}
\label{lem:rowapprox}
Let $(\x, \y)$ be a mixed strategy profile in a $k \times k$ bimatrix game $(R,
C)$. With probability at least $1 - \frac{2}{k}$ we can find an
$\epsilon$-approximate payoff vector for the row player using $O(\frac{k \cdot
\log k}{\epsilon^2})$ payoff queries.
\end{lemma}
\begin{proof}
We begin by giving a randomized method for finding the payoff of a fixed row $i
\in [k]$ with probability $1 - \frac{2}{k^2}$ using $\frac{8 \cdot \ln
k}{\epsilon^2}$ many payoff queries. We will make $T = \frac{8 \cdot \ln
k}{\epsilon^2}$ many payoff queries along row $i$, chosen according to the
probability distribution $\y$. Let $X_1, X_2, \dots, X_T$ be the result of
these queries, and define $\overline{X} = \frac{1}{T}(X_1 + X_2 + \dots +
X_T)$. Applying Hoeffding's inequality, and noting that all payoffs lie in the
range $[-1, 1]$, gives:
\begin{align*}
\pr(|\overline{X} - R_i \cdot \y_i| \ge \epsilon) & \le 2 \cdot \exp(-\frac{2
\cdot T \cdot \epsilon^2}{(1 + 1)^2})\\
& = 2 \cdot \exp(-2 \cdot \ln k)\\
& = \frac{2}{k^2}.
\end{align*}
Thus, with probability $1 - \frac{2}{k^2}$, we have that $\overline{X}$ is within $\epsilon$ of $(R \cdot \y)_i$. 

Now, to find an $\epsilon$-approximation of the row player's payoffs, we simply
apply the above method separately for each row $i \in [k]$. Clearly, this
will use $O(\frac{k \cdot \ln k}{\epsilon^2})$ many payoff queries. The
probability that we correctly compute an $\epsilon$-approximation of the row
player's payoffs is:
\begin{align*}
\left(1 - \frac{2}{k^2}\right)^k 
&\ge 1 - {k \choose 1} \cdot \frac{2}{k^2} \\
& = 1 - \frac{2}{k}.
\end{align*}
This completes the proof. 
\qed
\end{proof}

Note that, since we can swap the roles of the two players,
Lemma~\ref{lem:rowapprox} can also be used to find an $\epsilon$-approximate
payoff vector for the column player. Combining Lemma~\ref{lem:rowapprox} with
Theorem~\ref{thm:zerosum} gives the following corollary.

\begin{corollary}
\label{cor:zerosum}
Given a $k \times k$ zero-sum bimatrix game, with probability at least $(1 -
k^{-\frac{1}{8}})(1 - \frac{2}{k})^2$, we can compute an
$\epsilon$-Nash equilibrium $(\x, \y)$, and  $\epsilon$-approximate payoff
vectors for both players under $(\x, \y)$, using $O(\frac{k \cdot \log
k}{\epsilon^2})$ payoff queries.
\end{corollary}

We now introduce further notation for working with approximate payoff vectors.
Let $(\x, \y)$ be a mixed strategy profile in a bimatrix game $(R, C)$.
Let $\row$ and $\col$ be $\epsilon$-approximate payoff vectors for the row and
column player, respectively. We say that a row $i$ is a \emph{best response
according to $\row$} if $\row_i = \max_{j \in [k]} \row_j$, and that a column
$i$ is a best response according to $\col$ if $\col_i = \max_{j \in [k]}
\col_j$. We will frequently use the fact that, if $i$ is a best response
according to $\row$, and if $i'$ is an actual best response against $\y$, then: 
\begin{equation}
\label{eqn:bestresponse}
|\row_i - R_{i'} \cdot \y | \le \epsilon.
\end{equation}
This is because, if $\row_i > R_{i'} \cdot \y + \epsilon$, then $R_{i} \cdot \y
> R_{i'} \cdot \y$, which would contradict the fact that $i'$ is a best
response, and if
$\row_i < R_{i'} \cdot \y + \epsilon$ then $\row_i < \row_{i'}$, which would
contradict the fact that $i$ is a best response according to $\row$.
If row $i$ is a best response according to $\row$, then we
define the row player's \emph{regret according to $\row$} to be $\row_i - \x
\cdot \row$. Similarly, if column $i$ a best response according to $\col$, then
we define the regret according to $\col$ to be $\col_i - \col \cdot \y$. Note
that, if $i$ is an actual best response against $\y$, then we have:
\begin{align*}
| (\row - \x \cdot \row) - (R_i \cdot \y - \x \cdot R \cdot \y) | \le 2
\epsilon.
\end{align*}
In other words, the regret according to $\row$ is within $2 \epsilon$ of the
actual regret suffered by the row player under $(\x, \y)$.

\subsection{A randomized algorithm for finding an $\epsilon$-WSNE of a zero-sum
game}
\label{sec:grwsne}

In this section, we give a randomized query efficient algorithm for finding
approximate well-supported Nash equilibria in zero-sum games. Our approach is to
first compute an approximate Nash equilibrium of the zero-sum game using
Corollary~\ref{cor:zerosum}, and to then convert this into an $\epsilon$-WSNE
for the zero-sum game. Chen, Deng, and Teng have given an algorithm (henceforth
referred to as the CDT algorithm) that takes a $\frac{\epsilon^2}{8}$-Nash
equilibrium of a game, and in polynomial-time finds an $\epsilon$-WSNE for that
game~\cite{CDT09}. However, their method requires that we know the payoff of
every pure strategy in the $\frac{\epsilon^2}{8}$-Nash equilibrium. In our
setting, we only know approximate payoff vectors for both players, so the CDT
algorithm cannot be directly applied. In the next lemma, we show a variant of
their result, which can be applied when we only know approximate payoff vectors.

\begin{lemma}
\label{lem:wsneconvert}
Let $(R, C)$ be a bimatrix game, let $(\x, \y)$ be an
$\frac{\epsilon^2}{24}$-Nash equilibrium of $(R, C)$, and let $\row$ and $\col$
be $\frac{\epsilon^2}{24}$-approximate payoff vectors for the row and column
player under  $(\x, \y)$. Without making any payoff queries, we can construct an
$\epsilon$-WSNE of $(R, C)$.
\end{lemma}
\begin{proof}
Let $i^*$ be an actual best response against $\y$, and let $j^*$ be a best
response according to $\row$. We define the set $B = \{j \; : \; \row_{j^*} >
\row_{j} + \frac{\epsilon}{4}\}$, and our first task is to show that $B$
contains every row $j$ that is not a $\frac{\epsilon}{2}$-best response against
$\y$.

Note that $|\row_{j^*} - R_{i^*} \cdot \y| \le \frac{\epsilon^{2}}{24}$. 
Furthermore, note that for every row $j$ we have:
\begin{equation*}
R_{i^*} \cdot \y - R_j \cdot \y \le \row_{j^*} - \row_j  + \frac{\epsilon^2}{12}.
\end{equation*}
Therefore, if row $j$ satisfies $\row_{j^*} \le \row_{j} + \frac{\epsilon}{4}$
then:
\begin{align*}
R_{i^*} \cdot \y - R_j \cdot \y &\le \frac{\epsilon}{4} + \frac{\epsilon^2}{12} \\
&\le \frac{\epsilon}{2}. 
\end{align*}
This implies that $j$ is an $\frac{\epsilon}{2}$-best response against $\y$. So,
the set $B$ must contain every row $j$ that is not an $\frac{\epsilon}{2}$-best
response against $\y$.

Define the random variable $Y = R_{i^*} \cdot \y - \x \cdot R \cdot \y$ and the
random variable $X = \row_{j^*} - \x \cdot \row_j$. Note that, since $\x$ is a
$\frac{\epsilon^2}{24}$-Nash equilibrium, we know that the row player suffers
regret at most $\frac{\epsilon^2}{24}$ under $(\x, \y)$, and therefore we have
$E[Y] \le \frac{\epsilon^2}{24}$. Furthermore, since $X$ is the regret according
to $\row$, we know that:
\begin{align*}
E[X] &\le E[Y] + \frac{\epsilon^2}{12} \\
&=  \frac{\epsilon^2}{8}.
\end{align*}
By applying Markov's inequality, we obtain:
\begin{align*}
\pr(X \ge \frac{\epsilon}{4}) &\le \left. \frac{\epsilon^2}{8} \middle/
\frac{\epsilon}{4} \right. \\
& = \frac{\epsilon}{2}.
\end{align*}
Therefore, $\x$ must assign at most $\frac{\epsilon}{2}$ probability to rows in
$B$. We define $\x'$ to be a modification of $\x$ where all probability
assigned to rows in $B$ is is shifted arbitrarily to rows that are not in $B$. 

Now consider the column player. If we define $B' = \{i \; : \; \col_{i^*} >
\col_{i} + \frac{\epsilon}{4}\}$, then we can use the same argument as above to
prove that $B'$ contains every column that is not a $\frac{\epsilon}{2}$-best
response against $\x$, and that the column player assigns at most
$\frac{\epsilon}{2}$ probability to the columns in $B'$. Similarly, we define
$\y'$ to be a modification of $\y$ where all probability assigned to columns in
$B'$ is shifted arbitrarily to columns not in $B'$.

Note that every row in $\supp(\x')$ is a $\frac{\epsilon}{2}$-best response
against $\y$. Since $\y'$ differs from $\y$ by a shift of at most
$\frac{\epsilon}{2}$ probability, we have that every row in $\supp(\x')$ is an
$\epsilon$-best response against $\y'$. Using the same technique, we can argue
that every row in $\supp(\y')$ is an $\epsilon$-best response against $\x'$,
and therefore $(\x', \y')$ is an $\epsilon$-WSNE. \qed
\end{proof}

By using the reduction from Lemma~\ref{lem:wsneconvert}, it is now easy to see
that we can compute an $\epsilon$-WSNE of a zero-sum bimatrix game. The
following corollary is a combination of Corollary~\ref{cor:zerosum} and
Lemma~\ref{lem:wsneconvert}.

\begin{corollary}
\label{cor:wsne}
Given a $k \times k$ zero-sum bimatrix game, with probability at least $(1 -
k^{-\frac{1}{8}})(1 - \frac{2}{k})^2$, we can compute an
$\epsilon$-WSNE $(\x, \y)$ using $O(\frac{k \cdot \log
k}{\epsilon^4})$ payoff queries.
\end{corollary}

\section{A randomized algorithm for finding a $(\frac{3 - \sqrt{5}}{2} + \epsilon)$-Nash equilibrium}
\label{sec:enbm}

In this section, we present a randomized payoff-query efficient algorithm for
finding a $(\frac{3 - \sqrt{5}}{2} + \epsilon)$-Nash equilibrium in a bimatrix
game, where $\frac{3 - \sqrt{5}}{2} \approx 0.38197$. We adapt an algorithm of
Bosse, Byrka, and Markakis~\cite{BBM10} (henceforth referred to as the BBM
algorithm) for finding a $(\frac{3 - \sqrt{5}}{2} + \epsilon)$-Nash equilibrium
in a bimatrix game. The BBM algorithm solves a zero-sum game, and then makes a
decision based on the regret suffered by the players. We must adapt the
algorithm to work with approximate payoff vectors.

Let $(R, C)$ be a $k \times k$ bimatrix game, and define $D = R - C$. Let
$\alpha \in [0, 1]$ be a parameter that will be fixed later. The algorithm is as
follows.

\begin{enumerate}
\item Apply Theorem~\ref{thm:zerosum} for $\frac{\epsilon}{4}$ to the game $(D,
-D)$ in order to obtain $(\x, \y)$, which is a $\frac{\epsilon}{4}$-Nash
equilibrium. Then apply Lemma~\ref{lem:rowapprox} in order to find $\row$ and
$\col$, which are $\frac{\epsilon}{4}$-approximate payoff vectors for when $(\x,
\y)$  is played in the game $(R, C)$. This step succeeds with probability $(1 -
k^{-\frac{1}{8}})(1 - \frac{2}{k})^2$ and uses $O(\frac{k \cdot \log
k}{\epsilon^2})$ payoff queries.
\item We will assume, without loss of generality, that the regret according to
$\row$ is larger than the regret according to $\col$. Let row~$b$ be a best
response according to $\row$ in the game $(R, C)$, and let $g$ be the regret
according to $\row$. Since $b$ and $g$ are determined by $\row$, this step
requires no payoff queries.
\item Let $d$ be a best response for the column player against row $b$ in the
game $(R, C)$. This can be found using $k$ payoff queries, by querying every
column in row $b$.
\item If $g \le \alpha$, then output $(\x, \y)$. Otherwise, let $\delta =
\frac{1 - g}{2 - g}$ and output the following strategy, denoted as $(\hat{\x},
\hat{\y})$: the row player plays $b$ as a pure strategy, and the column player
plays $\y$ with probability $(1 - \delta)$ and $d$ with probability $\delta$.
This step uses no payoff queries.
\end{enumerate}

The following lemma shows that this algorithm is correct, in the case 
Theorem~\ref{thm:zerosum} 
 succeeds in finding an approximate Nash equilibrium
of $(D, -D)$, and Lemma~\ref{lem:rowapprox} succeeds in finding the approximate
payoff vectors in $(R, C)$.

\begin{lemma}
If $(\x, \y)$ is a $\frac{\epsilon}{4}$-Nash equilibrium of $(D, -D)$, and
$\row$ and $\col$ are $\frac{\epsilon}{4}$-approximate payoff vectors for $(\x,
\y)$ in $(R, C)$, then the algorithm outputs a $(\max(\alpha, \frac{1 -
\alpha}{2 - \alpha}) + \epsilon)$-Nash equilibrium of $(R, C)$.
\end{lemma}
\begin{proof}
First consider the case where $g \le \alpha$. 
Let $i$ be a best response for the
row player against $\y$. By applying Equation~\eqref{eqn:bestresponse}, we can 
obtain the following bound on the row player's regret:
\begin{align*}
R_i \cdot \y - \x \cdot R \cdot \y 
& \le \row_b - \x \cdot \row + \frac{\epsilon}{2} \\
& = g + \frac{\epsilon}{2} \\
&\le \alpha + \frac{\epsilon}{2}
\end{align*}
Thus, the row player's regret at most $\alpha + \frac{\epsilon}{2}$
under $(\x, \y)$. Since, by assumption, the regret according to $\col$ is at
most $g$, we can use the same argument to prove that the column player's regret
is at most $\alpha + \frac{\epsilon}{2}$. Therefore, $(\x, \y)$ is a
$(\alpha + \frac{\epsilon}{2})$-Nash equilibrium.

We now consider the case where $g > \alpha$. We first consider the row player.
Let $\hat{b}$ be a best response against $\hat{\y}$. The row player's regret
in $(\hat{\x}, \hat{\y})$ is:
\begin{align*}
R_{\hat{b}} \cdot \hat{\y} - R_b \cdot \hat{\y} &= (1 - \delta)(R_{\hat{b}} \cdot \y - R_b \cdot \y) + \delta \cdot (R_{\hat{b}, d} - R_{b, d}) \\
& \le (1 - \delta)(\row_{\hat{b}} - \row_b + \frac{\epsilon}{2}) + \delta \cdot (R_{\hat{b}, d} - R_{b, d}) \\
& \le (1 - \delta) \cdot \frac{\epsilon}{2} + \delta \cdot (R_{\hat{b}, d} - R_{b, d})
\\
& \le \frac{\epsilon}{2} + \delta.
\end{align*}
The first inequality holds because $\row$ is is an
$\frac{\epsilon}{4}$-approximate payoff vector against~$\y$. The second
inequality holds because $b = \max_{i \in [k]} \row_i$. The third inequality
holds because $1 - \delta \le 1$ and $R_{i, j} \in [0, 1]$ for all $i$ and $j$.

So, we have shown that the row player's regret is at most 
\begin{align*}
\delta + \frac{\epsilon}{2} &= \frac{1 - g}{2 - g} + \frac{\epsilon}{2} \\
& \le \frac{1 - \alpha}{2 - \alpha} + \frac{\epsilon}{2}
\end{align*} 
The inequality above holds because $\frac{1 - g}{2 - g}$ is a decreasing
function when $g \le 1$, and because
$g > \alpha$.  

We now consider the column player. Since $(\x, \y)$ is an
$\frac{\epsilon}{4}$-Nash equilibrium of $(D, -D)$, the row player's regret
under $(\x, \y)$ in $(D, -D)$ must be at most $\frac{\epsilon}{4}$. Thus, all
rows $i$ must satisfy $D_i \cdot \y \le \x \cdot D \cdot \y +
\frac{\epsilon}{4}$. Applying this inequality for row $b$ gives the following:
\begin{align}
\nonumber
D_b \cdot \y &\le \x \cdot D \cdot \y + \frac{\epsilon}{4} \\
\nonumber
(R - C)_b \cdot \y &\le \x \cdot (R - C) \cdot \y + \frac{\epsilon}{4} \\
\label{eqn:up}
 R_b \cdot \y - \x \cdot R \cdot \y + \x \cdot C \cdot \y - \frac{\epsilon}{4}
&\le C_b \cdot \y 
\end{align}
Since $\row$ is a $\frac{\epsilon}{4}$-approximate payoff vector we have:
\begin{align*}
R_b \cdot \y - \x \cdot R \cdot \y &\ge \row_b - \x \cdot \row -
\frac{\epsilon}{2} \\ 
&= g - \frac{\epsilon}{2} 
\end{align*}
Substituting this into Equation~\eqref{eqn:up}
\begin{align}
\nonumber
C_b \cdot \y
\nonumber
& \ge g + \x \cdot C \cdot \y  - \frac{3}{4}\epsilon \\
\label{eqn:lol}
& \ge g - \epsilon
\end{align}

Recall that $d$ is an actual best response against~$b$ for the column player.
Since $\hat{\x}$ plays $b$ as a pure strategy, we have therefore have that $d$
is a best response against $\hat{\x}$. Moreover, observe that
\begin{align} 
\nonumber
(1-\delta)(1 - g) &= \frac{(1 - g)(2-g) + (1-g)^2}{2-g} \\
\label{eqn:down}
&= \frac{1-g}{2-g}.
\end{align}
Thus, the column player's regret when playing $\hat{\y}$
against $\hat{\x}$ is:
\begin{align*}
(\hat{\x} \cdot C)_{d} - \hat{\x} \cdot C \cdot \hat{\y} &= 
C_{b, d} - C_{b} \cdot \hat{\y} \\
& = C_{b, d} - ((1 - \delta) \cdot C_b \cdot \y + \delta \cdot C_{b, d}) \\
& = (1 - \delta)(C_{b, d} - C_b \cdot \y) \\
& \le (1 - \delta)(1 - g + \epsilon)  & \text{[By Equation~\eqref{eqn:lol}]}\\
&\le \frac{1 - g}{2 - g} + \epsilon & \text{[By Equation~\eqref{eqn:down}]} \\
&< \frac{1 - \alpha}{2 - \alpha} + \epsilon.
\end{align*}
We have now shown that, in the case where $g > \alpha$, both players have regret
at most $\frac{1 - \alpha}{2 - \alpha} + \epsilon$, which implies that
$(\hat{x}, \hat{y})$ is a $(\frac{1 - \alpha}{2 - \alpha} + \epsilon)$-Nash
equilibrium. \qed
\end{proof}

As shown by Bosse, Byrka, and Markakis, we have that the expression
$\max(\alpha, \frac{1 - \alpha}{2 - \alpha})$ is minimized when $\alpha =
\frac{3 - \sqrt{5}}{2} \approx 0.38197$. Thus, we have the
following theorem.

\begin{theorem}
Given a $k \times k$ bimatrix game, with probability at least $(1 -
k^{-\frac{1}{8}})(1 - \frac{2}{k})^2$, we can compute a $(\frac{3 -
\sqrt{5}}{2} + \epsilon)$-Nash equilibrium using $O(\frac{k \cdot \log
k}{\epsilon^2})$ payoff queries.
\end{theorem}

\section{A randomized algorithm for finding finding a
$(\frac{2}{3}+\epsilon)$-well-supported Nash equilibrium}
\label{sec:wsnebm}

In this section, we give a randomized $O(\frac{k \cdot \log k}{\epsilon^4})$
payoff query algorithm for finding a $(\frac{2}{3} + \epsilon)$-WSNE equilibrium
in a bimatrix game, and we give a randomized $O(\frac{k \cdot \log
k}{\epsilon^4})$ payoff query algorithm for finding a $(\frac{1}{2} +
\epsilon)$-WSNE in a zero-one bimatrix game. We follow the algorithm of
Kontogiannis and Spirakis
(henceforth referred to as the KS algorithm) for finding a $\frac{2}{3}$-WSNE in
a bimatrix game~\cite{KS10}. Their approach can be summarised as follows: they
first perform a preprocessing step in which they check whether the game has a
pure $\frac{2}{3}$-WSNE. If it does not, then they construct a zero-sum game and
compute an exact Nash equilibrium for it. They show that, if the original game
does not have a pure $\frac{2}{3}$-WSNE, then an exact Nash equilibrium in the
zero-sum game is a $\frac{2}{3}$-WSNE in the original game.

There are two problems with this approach in the payoff query setting. Firstly,
it has been shown that finding an exact Nash equilibrium of a
$k \times k$ zero-sum game requires $k^2$ queries~\cite{FGGS13}.
To solve this problem, we substitute an
$\epsilon$-WSNE of the zero-sum game in place of an exact Nash equilibrium, and
we show that, after this substitution, the  KS algorithm will still produce a
well-supported Nash equilibrium of the original game. The second problem is that
we are unable to perform the preprocessing step in a query-efficient manner. A
naive algorithm would require $k^2$ payoff queries in order to verify whether
there is a pure $\frac{2}{3}$-WSNE, and it is not clear how this can be
improved. To solve this problem, we show that the preprocessing step does not
need to be carried out before the zero-sum game has been solved. Instead, we
find an $\epsilon$-WSNE in the zero-sum game, and then check whether it is a
$\frac{2}{3}$-WSNE in the original. If it is not, then we show how this
information can be used to find a pure $(\frac{2}{3} + \epsilon)$-WSNE in the
original game.

Let $(R, C)$ be a bimatrix game, where~$R$ is the payoff matrix of the row
player, and~$C$ is the payoff matrix of the column player. The KS algorithm uses
the following definitions.
\begin{align*}
D := \frac{1}{2}(R - C) && X := -\frac{1}{2}(R + C)
\end{align*}
Observe that $D = R + X$. The KS algorithm finds an exact Nash equilibrium of
the zero-sum game $(D, -D)$. In contrast to this, we do the following:
\begin{itemize}
\item we will apply Corollary~\ref{cor:wsne} to $(D, -D)$ to obtain an
$\frac{\epsilon}{10}$-WSNE for $(D, -D)$, which we denote as $(\x, \y)$.
\item We then obtain approximate payoff vectors for when $(\x, \y)$ is played in
$(R, C)$: we apply
Lemma~\ref{lem:rowapprox} to obtain $\frac{\epsilon}{10}$-approximate payoff
vectors $\row$, which approximates $R \cdot \y$, and $\col$, which approximates
$\x \cdot C$.
\end{itemize}
These two steps succeed with probability at least $(1 -
k^{-\frac{1}{8}})(1 - \frac{2}{k})^3$. We will fix $(\x, \y)$, $\row$,
and $\col$ for the rest of this section.

In the following lemma, we show how the preprocessing of the KS algorithm can be
delayed until after the zero-sum game has been solved. In particular, we show
that if there is a row in the support of $\x$ that is far from being an
approximate best response, then, in a query-efficient manner, we can find a pure
strategy profile $i,j \in [k]$ such that $(R + C)_{i,j}$ is large. The parameter
$z$ will allow us to apply this lemma for both the zero-one and general bimatrix
games: in the zero-one case we will set $z = 0.5$, and for general bimatrix
games we will set $z = \frac{2}{3}$.

\begin{lemma}
\label{lem:twoz}
Let $z \in [0, 1]$. If there is an $i' \in \supp(\x)$ and $i \in [k]$ such that
$\row_i > \row_{i'} + z - \frac{\epsilon}{5}$, then using $k$ payoff queries, we
can find a $j \in [k]$ such that $(R + C)_{i, j} > 2z - \epsilon$.
\end{lemma}
\begin{proof}
Since $\row$ is an $\frac{\epsilon}{10}$-approximate payoff vector for the row
player, and since 
$\row_i > \row_{i'} + z - \frac{\epsilon}{5}$
we have:
\begin{equation}
\label{eqn:one}
R_i \cdot \y > R_{i'} \cdot \y + z - \frac{2}{5} \epsilon. 
\end{equation}
Since $\x$ is an $\frac{\epsilon}{10}$-WSNE in $(D, -D)$, and since $i' \in
\supp(\x)$ by assumption, we have:
\begin{align*}
D_{i'} \cdot \y &\ge D_{i} \cdot \y - \frac{\epsilon}{10} \\
(R + X)_{i'} \cdot \y &\ge (R + X)_{i} \cdot \y - \frac{\epsilon}{10} \\
R_{i'} \cdot \y &\ge R_{i} \cdot \y - (X_{i'} - X_{i}) \cdot \y - \frac{\epsilon}{10}.
\end{align*}
Combining this with Equation~\eqref{eqn:one} yields:
\begin{align*}
R_{i'} \cdot \y &> R_{i'} \cdot \y + z - \frac{2\epsilon}{5} - (X_{i'} - X_{i}) \cdot
\y  - \frac{\epsilon}{10}\\
(X_{i'} - X_{i}) \cdot \y &> z - \frac{\epsilon}{2}.
\end{align*}
Since $X = -\frac{1}{2}(R + C)$, we have that $X_{i'} \cdot \y \in [-1, 0]$ and
hence $X_{i'} \cdot \y \le 0$. Therefore, we have:
\begin{align*}
- X_{i} \cdot \y &> z - \frac{\epsilon}{2} \\
\frac{1}{2}(R + C)_i \cdot \y &> z - \frac{\epsilon}{2} \\
(R + C)_i \cdot \y &> 2z - \epsilon.
\end{align*}
In order for this inequality to hold, there must be at least one column $j$
such that $(R + C)_{i, j} > 2z - \epsilon$. 
Since we know row $i$, we can find column $j$ using $k$ payoff queries. \qed
\end{proof}

Note that, by swapping the roles of the two players, Lemma~\ref{lem:twoz} can
also be applied for the column player. We can now prove the main result of this
section. 

\begin{theorem}
Let $(R, C)$ be a $k \times k$ bimatrix game. With probability at least $(1 -
k^{-\frac{1}{8}})(1 - \frac{2}{k})^3$ and using 
$O(\frac{k \cdot \log k}{\epsilon^4})$ we can compute
a $(\frac{2}{3} + \epsilon)$-WSNE if $(R, C)$ is a general bimatrix game,
or
 a $(\frac{1}{2} + \epsilon)$-WSNE if $(R, C)$ is a zero-one bimatrix game.
\end{theorem}
\begin{proof}
As we have described, the algorithm spends
$O(\frac{k \cdot \log k}{\epsilon^4})$ in order to compute $(\x, \y)$, $\row$, and
$\col$.
We first prove the result for general bimatrix games.

Observe that, since since $\row$ is an $\frac{\epsilon}{10}$-approximate payoff
vector, if we have $\row_i \le \row_{i'} + \frac{2}{3} - \frac{\epsilon}{5}$ for
two rows $i, i' \in [k]$, then we have $R_i \cdot \y \le R_{i'} + \frac{2}{3}$.
So, our algorithm will 
check whether: 
\begin{itemize}
\item $\row_{i} \le \row_{i'} + \frac{2}{3} -
\frac{\epsilon}{5}$ for every $i \in [k]$ and every $i' \in \supp(\x)$, and
\item 
$\col_{i} \le \col_{i'} + \frac{2}{3} - \frac{\epsilon}{5}$ for every $i
\in [k]$ and every $i' \in \supp(\y)$. 
\end{itemize}
If both of these checks succeed, then $(\x, \y)$ is a $\frac{2}{3}$-WSNE in $(R,
C)$, and we are done. On the other hand, if either of the two checks fail, then
we can apply Lemma~\ref{lem:twoz} to obtain a pair $i, j \in [k]$ such that $(R
+ C)_{i, j} > \frac{4}{3} - \epsilon$. This implies that both $R_{i,j} >
\frac{1}{3} - \epsilon$ and $C_{i, j} > \frac{1}{3} - \epsilon$, and therefore
$i$ and $j$ form a pure $(\frac{2}{3} + \epsilon)$-WSNE. 

For zero-one games, the proof is similar. The algorithm will check whether
$\row_{i} \le \row_{i'} + (\frac{1}{2} + \epsilon) - \frac{\epsilon}{5}$ for
every $i \in [k]$ and every $i' \in \supp(\x)$, and whether $\col_{i} \le
\col_{i'} + (\frac{1}{2} + \epsilon) - \frac{\epsilon}{5}$ for every $i \in [k]$
and every $i' \in \supp(\y)$. If both of these checks succeed, then we have that
$(\x, \y)$ is a $(\frac{1}{2} + \epsilon)$-WSNE. Otherwise, we can apply
Lemma~\ref{lem:twoz} to obtain a pair $i, j \in [k]$ such that $(R + C)_{i, j} >
1 + 2 \epsilon - \epsilon > 1$. However, since this is a zero-one game, the only
way to have $(R + C)_{i, j} > 1$ is if $R_{i, j} = 1$ and $C_{i, j}  = 1$.
Therefore $(i, j)$ is a pure Nash equilibrium. 
\qed
\end{proof}

\section{Conclusion}

In this paper, we have given a complete characterization of the deterministic
query complexity of $\epsilon$-Nash equilibria for constant $\epsilon$, we have
given randomized upper bounds for both $\epsilon$-Nash equilibria and
$\epsilon$-WSNE, and we have initiated work on lower bounds against randomized
algorithms for finding $\epsilon$-Nash equilibria. 

There are many open problems arising from this work. In addition to ruling out
query-efficient randomized algorithms for finding exact Nash equilibria, our
lower bound in Section~\ref{sec:rlower} also rules out query efficient
algorithms for finding $\frac{1}{6k}$-Nash equilibria in $k \times k$
bimatrix games. The most obvious open problem stemming from this is to prove a
lower bound for randomized algorithms for constant approximations. At the very
least, it would be desirable to have matching $\Omega(k \cdot \log k)$ lower
bounds for our algorithms in Sections~\ref{sec:enbm} and~\ref{sec:wsnebm}. Also,
does there exist a constant $\epsilon < \frac{3 - \sqrt{5}}{2}$ for
which the randomized query complexity is $\omega(k \cdot \log k)$, or a constant
$\epsilon < \frac{2}{3}$ for which the randomized query complexity of finding an
$\epsilon$-WSNE is $\omega(k \cdot \log k)$?

Our result in Section~\ref{sec:epsdlower} completes the characterization of the
deterministic query complexity for $\epsilon$-Nash equilibria for constant
$\epsilon$, but we do not have such a characterization for $\epsilon$-WSNE. Of
course, since every $\epsilon$-WSNE is an $\epsilon$-Nash equilibrium, the lower
bound of Section~\ref{sec:epsdlower} applies for $\epsilon < \frac{1}{2}$, but
we do not know much about the deterministic query complexity of
$\epsilon$-WSNE with $\epsilon \ge \frac{1}{2}$. One easy initial observation on
this topic is the following lemma, proved in Appendix~\ref{app:easy}, which
shows that for every $\epsilon < 1$, the deterministic query complexity of
finding an $\epsilon$-WSNE is $\Omega(k)$. It has been shown that we can always
find a $(1 - \frac{1}{k})$-Nash equilibrium using no queries at
all~\cite{FGGS13}, so this lemma shows that we should expect $\epsilon$-WSNE
to behave differently to $\epsilon$-Nash equilibria when $\epsilon >
\frac{1}{2}$.

\begin{lemma}
\label{lem:shitty}
For every $\epsilon < 1$, the deterministic query complexity of finding an
$\epsilon$-WSNE in a $k \times k$ bimatrix game is at least $k - 1$, even in
zero-one games.
\end{lemma}

In Sections~\ref{sec:enbm} and~\ref{sec:wsnebm}, we adapt two polynomial-time
approximation algorithms to give randomized query-efficient protocols. However,
in both cases, we do not use the best possible polynomial-time approximations.
There is a polynomial time algorithm that finds a $0.3393$-Nash
equilibrium~\cite{TS08} using a gradient descent method, but it is not at all
clear how this method could be implemented in a query efficient manner. For
WSNE, there is a polynomial-time algorithm for finding a $(\frac{2}{3} -
0.004735)$-WSNE~\cite{FGSS12}, which uses an \emph{exhaustive} search of all $2
\times 2$ supports in order to improve the KS-algorithm. Even if this search
could be implemented in a query efficient manner, this would only give a tiny
improvement over our result in Section~\ref{sec:wsnebm}.

\paragraph{\bf Acknowledgement} We would like to thank Paul Goldberg and Noam
Nisan for useful discussions on this topic.

\bibliographystyle{acmsmall}
\bibliography{references}

\appendix
\newpage

\section{A linear lower bound for finding $\epsilon$-WSNE when $\epsilon < 1$}
\label{app:easy}

In this section, we show a linear lower bound for finding any non-trivial
well-supported Nash equilibrium in a zero-one game. More precisely, we will
show that, for any $\epsilon < 1$, all algorithms must make at least $k-1$
payoff queries in order to find an $\epsilon$-WSNE in a $k \times k$ bimatrix
game. 

To prove our result, we will assume that all payoff queries return payoff~$0$
for both the row and column player. We will show that, when all queries are
responded to in this way, all algorithms must make at least $k-1$ payoff
queries in order to correctly determine an $\epsilon$-WSNE. We first show the
following lemma.

\begin{lemma}
\label{lem:wsnequery}
Let $(\x, \y)$ be a $\epsilon$-WSNE. Let $r$ be a row that is played with
probability strictly less than $1$ in $\x$. At most $\epsilon$ probability can
be assigned to columns in $r$ that have not been queried.
\end{lemma}
\begin{proof}
Suppose, for the sake of contradiction, that $\x$ assigns strictly more than
$\epsilon$ probability to unqueried columns in $r$. Let $U = \{c \; : \; (r, c)
\text{ did not receive a payoff query}\}$. We construct a row player payoff
matrix $R$ as follows:
\begin{equation*}
R_{i,j} = \begin{cases}
1  & \text{if $i = r$ and $j \in U$,}\\
0 & \text{otherwise.}
\end{cases}
\end{equation*}
This matrix is consistent with all payoff queries that have been made so far.
Since $\x$ assigns strictly more than $\epsilon$ probability to the columns in
$U$, the payoff of row $r$ is strictly greater than $\epsilon$. Moreover, the
payoff of every row $i \ne r$ is $0$. Since $r$ is not played with probability
$1$ by $\x$, some probability must be assigned to a row $r'$ with payoff $0$.
Therefore, we have:
\begin{equation*}
R_{r'} \cdot \y - R_r \cdot \y > \epsilon - 0.
\end{equation*} 
This proves that row $r$ is not an $\epsilon$-best response against
$\y$, which provides our contradiction. \qed
\end{proof}

Having shown Lemma~\ref{lem:wsnequery}, we can now provide the proof of
Lemma~\ref{lem:shitty}.

\begin{proof}[of Lemma~\ref{lem:shitty}]
Let $(\x, \y)$ be an $\epsilon$-WSNE. Let $W$ be the set of rows that are not
played with probability $1$ by $\x$. Note that $W$ contains at least $k-1$
rows. By Lemma~\ref{lem:wsnequery}, in each row in $W$, there must be at least
$1 - \epsilon$ probability assigned to queried columns. Since $\epsilon < 1$,
this implies that each row in $W$ must have at least one queried column,  which
implies that we must have made at least $k-1$ queries. \qed
\end{proof}

\end{document}